\documentclass{amsart}

\usepackage{a4wide}
\usepackage{amssymb}
\usepackage{amsmath}
\usepackage{amsfonts}
\usepackage{amsthm}
\usepackage{graphicx}
\usepackage{epsfig}
\usepackage{longtable}
\usepackage[nocompress]{cite}

\usepackage{hyperref}
\usepackage[usenames,dvipsnames]{color}%
\newtheorem{thm}{Theorem}[section]
\newtheorem{lem}[thm]{Lemma}
\newtheorem{cor}[thm]{Corollary}

\theoremstyle{definition}

\theoremstyle{remark}
\newtheorem*{rmk}{Remark}

\numberwithin{equation}{section}

\newcommand{\DEF}{{:=}}

\newcommand{\Cset}{\mathbb{C}}

\newcommand{\re}{\mathop{\mathrm{Re}}}

\newcommand{\dd}{\,d}

\DeclareMathOperator{\bernoulliB}{B}
\DeclareMathOperator{\BigOh}{\mathcal{O}}

\DeclareMathOperator{\CAP}{cap}

\DeclareMathOperator{\gammafcn}{\Gamma}

\DeclareMathOperator{\GegenbauerC}{\mathrm{C}}

\DeclareMathOperator{\digammafcn}{\psi}

\DeclareMathOperator{\LegendreP}{P}

\DeclareMathOperator{\zetafcn}{\zeta}

\newcommand{\Pochhsymb}[2]{{\left(#1\right)_{#2}}}

\allowdisplaybreaks[1]

\title[Complete minimal logarithmic energy asymptotics of the logarithmic energy \dots]{Complete minimal logarithmic energy asymptotics for points in a compact interval: a consequence of the discriminant of Jacobi polynomials} 
\author{ J. S. Brauchart} 
\thanks{
\noindent 
The research of this author was supported, in part, by the Austrian Science Fund FWF project F5510 (part of the Special Research Program (SFB) ``Quasi-Monte Carlo Methods: Theory and Applications'') and M2030 Meitner-Programm ``Self organization by local interaction''. 
}

\date{\today}

\begin{document}

\address{J. S. Brauchart:
Institute of Analysis and Number Theory, 
Graz University of Technology,
Kopernikusgasse 24/II, 
8010 Graz, 
Austria}
\email{j.brauchart@tugraz.at}

\begin{abstract} 
The electrostatic interpretation of zeros of Jacobi polynomials, due to Stieltjes and Schur, enables us to obtain the complete asymptotic expansion as $n \to \infty$ of the minimal logarithmic potential energy of $n$ point charges restricted to move in the interval $[-1,1]$ in the presence of an external field generated by endpoint charges. By the same methods, we determine the complete asymptotic expansion of the logarithmic energy $\sum_{j\neq k} \log(1/| x_j - x_k |)$ of Fekete points, which, by definition, maximize the product of all mutual distances $\prod_{j\neq k} | x_j - x_k |$ of $N$ points in $[-1,1]$ as $N \to \infty$. The results for other compact intervals differ only in the quadratic and linear term of the asymptotics. 
%
%
%
Explicit formulas and their asymptotics follow from the discriminant, leading coefficient, and special values at $\pm 1$ of Jacobi polynomials. For all these quantities we derive complete Poincar{\'e}-type asymptotics.
%
\end{abstract}

\subjclass[2000]{}

\maketitle

\section{Introduction and statement of results}

Point sets characterized by means of minimizing a suitably defined potential energy function have applications in a surprising number of problems in various fields of science and engineering ranging from physics over chemistry to geodesy and mathematics. We refer the reader to \cite{KuSa1997,HaSa2004,CoKu2007,BrDi2013,BrGr2015,Le2016,SaSe2015,SaSe2015b,RoSe2016,PeSe2015,BeMaOr2016,Be2015,BoDrHaSaSt2017,CoKuMiRaVi2017,Vi2017,Sk2019,BeDrOr2020,BeKnNo2020,PeSe2020,Ba2021,Pau2021,BeEtMaOr2021,BiGlMaPaVl2021,KivanMe2021} and the book \cite{BoHaSaBook2019}. A fundamental question concerns the asymptotic expansion of the minimal energy as the number of points tend to infinity. In general, at best only one or two terms are known; cf. \cite{LoSa2010,BrHaSa2012b,BeEt2020,BeSa2018,LoMcC2021} in case of the sphere and \cite{MaMaRa2004,Bo2012} for curves. A notable exception are the minimal energy asymptotics for the unit circle for a whole class of energy functionals for which equally spaced points are optimal configurations. In these cases the energy formula can be written in a form that provides a complete asymptotic expansion in terms of powers of the number of points (see \cite{Br2016,BrHaSa2012,BrHaSa2009}): for $s \in (-2,\infty)$ with $s \neq 0, 1, 3, 5, \dots$ and for every $p = 1, 2, 3, \dots$, one has for the optimal \emph{Riesz $s$-energy} the asymptotic expansion
\begin{equation*} 
\begin{split}
\mathcal{L}_s(N) 
&= W_s \, N^2 + \frac{2\zetafcn(s)}{(2\pi)^s} \, N^{1+s} + 
\sum_{n=1}^p \alpha_n(s) \frac{2\zetafcn(s-2n)}{(2\pi)^s} \, N^{1+s-2n} + \BigOh_{s,p}(N^{-1+ s-2p}) 
\end{split}
\end{equation*}
as $N \to \infty$, where the constant $W_s$ is explicitly known, $\zetafcn(s)$ is the classical Riemann zeta function, and the explicitly computable coefficients $\alpha_n(s)$, $n\geq0$, satisfy the generating function
relation
\begin{equation*} \label{sinc.power.0}
\left( \frac{\sin \pi z}{\pi z} \right)^{-s} = 
\sum_{n=0}^\infty \alpha_n(s) \, z^{2n}, \quad |z|<1, \ s\in \Cset.
\end{equation*}
The logarithmic energy of $N$ equally spaced points, which provide minimizing configurations, simply is
\begin{equation*}
\mathcal{L}_0(N) = - N \, \log N.
\end{equation*}
We remark that for general curves much less is known. We refer to \cite{MaMaRa2004,Bo2012}. In the following we shall utilize the fact that zeros of classical orthogonal polynomials can be characterized as minimizing configurations of certain potential energy functions for logarithmic point interactions. This approach enables us to derive complete asymptotic expansions. 

Let $A$ be an infinite compact subset of the complex plane $\mathbb{C}$. A configuration of $N$ points $\zeta_1,\dots,\zeta_N \in A$, $N\geq2$, that maximizes the product of all mutual distances $\prod_{j\neq k}|z_j-z_k|$ among $N$-point systems $z_1,\dots,z_N \in A$ is called an {\em $N$-th system of Fekete points of $A$}. The maximum
\begin{equation} \label{eq:Nth.discriminant.of.interval}
\Delta_N(A) \DEF \max_{z_1,\dots,z_N \in A} \mathop{\prod_{j=1}^N \prod_{k=1}^N}_{j \neq k} \left| z_j - z_k \right|
\end{equation}
is the {\em $N$-th discriminant of $A$}. A fundamental potential-theoretic result for the {\em transfinite diameter or logarithmic capacity} $\CAP A$  of $A$ is 
\begin{equation*}
\CAP A = \lim_{N\to\infty} \left[ \Delta_N(A) \right]^{1/[N(N-1)]}.
\end{equation*}
%
Fekete points, by definition, are points that maximize the Vandermonde determinant that appears in the polynomial Lagrange interpolation formula. 
It was Fekete~\cite{Fe1926} who investigated the connection between polynomial interpolation and the discrete logarithmic energy problem, which for given $N$ consists of finding those $N$-point configurations with minimal discrete {\em logarithmic energy} 
\begin{equation} \label{eq:log.energy}
E_0(z_1,\dots,z_N) \DEF \mathop{\sum_{j=1}^N \sum_{k=1}^N}_{j \neq k} \log \frac{1}{\left| z_j - z_k \right|}, \qquad z_1,\dots,z_N\in A.
\end{equation}
We define the {\em logarithmic $N$-point energy of $A$} to be
\begin{equation} \label{eq:min.log.energy}
\mathcal{E}_0(A; N) \DEF \sup \left\{ E_0(z_1,\dots,z_N) : z_1,\dots,z_N\in A \right\} = - \log \Delta_N(A).
\end{equation}
One main goal of this paper is to derive the complete asymptotic expansion of $\mathcal{E}_0(A; N)$ as $N\to\infty$ when $A$ is the interval $[-1,1]$; see Theorem~\ref{thm:log.n.point.energy.asymptotics}. Indeed, regarding line-segments, it suffices to consider the interval $[-1,1]$, since the $N$-th discriminant of the rotated, dilated, and translated set $A^\prime=a + \eta e^{i \phi} A$ is given by $\Delta_N(A^\prime) = \eta^{N(N-1)} \Delta_N(A)$ and, therefore, $\mathcal{E}_0(A^\prime; N) - \mathcal{E}_0(A; N) = - ( \log \eta ) N ( N - 1 )$.

Let $q > 0$ and $p > 0$ be numbers representing charges at the left endpoint and right endpoint, respectively, of the interval $[-1,1]$. The problem of finding $n$ points $x_1^{(n)},\dots,x_n^{(n)}$, the locations of unit point charges, in the interior of $[-1,1]$ such that the expression
\begin{equation} \label{Jacobi.expr}
T_n(x_1,\dots,x_n) \DEF \prod_{i=1}^n \left( 1 - x_i \right)^p \, \prod_{j < k} \left| x_j - x_k \right| \, \prod_{\ell=1}^n \left( 1 + x_{\ell} \right)^q
\end{equation}
is maximized, or equivalently, $\log(1/T_n)$ is minimized over all $n$-point systems $x_1,\dots,x_n$ in $[-1,1]$, is a classical problem that owes its solution to Stieltjes \cite{St1884,St1885} (also see Schur~\cite{Sch1918}).
%
In analogy to the $N$-th discriminant of a compact set $A$ we may define the \emph{$n$-th $(p,q)$-discriminant of $[-1,1]$} as
\begin{equation} \label{eq:Nth.pq.discriminant.of.interval}
\Delta_n^{(p,q)}([-1,1]) \DEF \max_{x_1,\dots,x_n \in [-1,1]} \left( T_n(x_1,\dots,x_n) \right)^2.
\end{equation}
The quantity $\log(1/T_n^2)$ can be interpreted as the potential energy of the point charges at $x_1,\dots,x_n$ in an external field exerted by the charge $p$ at $x=1$ and the charge $q$ at $x=-1$, where the 'points' interact according to a logarithmic potential. We shall call such minimal potential energy points {\em {\bf elliptic} Fekete points} in order to distinguish them from the Fekete points defined previously. Stieltjes showed that the points
$x_1^{(n)},\cdots,x_n^{(n)}$
of {minimal potential energy} are, in fact, the zeros of the Jacobi polynomial $P_n^{(\alpha,\beta)}$, where $\alpha=2p-1$ and $\beta=2q-1$. 
A more modern approach is to have external fields in form of appropriate weight functions instead of constraints. (See, e.g., \cite{Is2000b} for a discussion of this model.) We also refer the interested reader to the survey article \cite{MaMaMa2007}.

Stieltjes’ ingenious observation that the zeros of classical orthogonal polynomials have an electrostatic interpretation in terms of logarithmic potential enables us to find, for every $n\geq2$, the explicit elliptic Fekete $n$-point configuration for the discrete logarithmic energy problem associated with the given family of orthogonal polynomials.
Moreover, since the target functions of the respective maximum problems are closely related to the discriminants of the classical polynomials, the asymptotic expansion of the potential energy of elliptic Fekete $n$-point configurations as $n \to \infty$ can be obtained. Our goal is to derive the complete asymptotic expansion of the potential energy of elliptic Fekete $n$-point configurations associated with the external field problem induced by classical orthogonal polynomials. 

We remark that the approach used here can be also applied to point systems in $[0,\infty)$ and $\mathbb{R}$ with suitable constraints on the centroid or inertia of the point system which leads to the study of zeros of Laguerre and Hermite polynomials, respectively. A generalization are so-called Menke systems for the real line studied in \cite{MaBrSa2009}; see also \cite{Me1972, Me1974}. Such systems consist of two interlaced sets of points which can be characterized as zeros or extrema of classical orthogonal polynomials. The asymptotic analysis of the associated discriminants is technically much more involved and we leave the presentation of these results to follow up papers. 

%
%

{\bf Outline of the paper:} In the remaining part of the introduction we present the asymptotic expansions for elliptic Fekete points in the interval $[-1,1]$ and compare the results with the expansion for Fekete points. In Section~\ref{sec:aux.res}, we gather asymptotic results for the discriminant of the Jacobi polynomial. The proofs of the main asymptotic results are presented in Section~\ref{sec:proofs.main.res}. The Appendix collects technical results that are frequently used in the asymptotic analysis.

\subsection{Preliminaries} 
Our asymptotic expansions are of Poincar{\'e}-type and we adapt the notion of writing them as infinite series (even if an infinite series does not converge). We make use of the usual computational rules. 
%
The coefficients of the asymptotic expansions will be given in terms of the Riemann zeta function~$\zetafcn(s)$ and the Hurwitz zeta function $\zetafcn(s,a)$ and their (partial) derivatives with respect to $s$ evaluated at negative integers $s$. The well-known relation
\begin{equation*}
\zetafcn( -m, a ) = - \frac{\bernoulliB_{m+1}( a )}{m+1}, \qquad m \in \mathbb{N}_0,
\end{equation*}
enables us to use instead Bernoulli polynomials $\bernoulliB_m$ and the Bernoulli numbers $B_m$. 
The {\em Glaisher-Kinkelin} constant (see \cite[p.~135]{Fi2003})\footnote{The established symbol for the Glaisher-Kinkelin constant is $A$ which we also use for a generic compact set. The use of the symbol $A$ should be clear from the context.} is defined by
\begin{equation} \label{eq:Glaisher-Kinkelin}
A \DEF \lim_{n\to \infty} \frac{1^1 2^2 \cdots n^n}{n^{n(n+1)/2+1/12} e^{-n^2/4}} = 1.28242712\dots.
\end{equation}
and appears in our computations by means of the well-known relation $\zetafcn^\prime(-1) = 1/12 - \log A$. The {\em polygamma} function is defined by \cite[6.4.1]{AbSt1992}
\begin{equation*}
\digammafcn^{(n)}(z) = \frac{\dd^{n+1}}{\dd z^{n+1}} \log \gammafcn(z), \qquad n = 1,2,3, \dots.
\end{equation*}
Using Liouville's fractional integration and differentiation operator, one can also define polygamma functions of negative order (called ``negapolygammas'' in \cite{Go1995}) as (see \cite{Ad1998})
\begin{equation*}
\digammafcn^{(-n)}(z) \DEF \frac{1}{(n-2)!} \int_0^z \left( z - t \right)^{n-2} \log \gammafcn(t) \dd t, \quad \re z > 0, \qquad n=1,2,3,\dots.
\end{equation*}
Since 
\begin{equation} \label{eq:negadigamma}
\digammafcn^{(-2)}(x) = \int_0^x \log \gammafcn(t) \dd t = \frac{\left(1-x\right)x}{2} + \frac{x}{2} \log 2\pi - \zetafcn^\prime(-1) + \left. \frac{\partial}{\partial s} \zetafcn(s,x) \right|_{s=-1},
\end{equation}
we rewrite $\zetafcn^\prime(-1,x) \DEF \frac{\partial}{\partial s} \zetafcn(s,x) \big|_{s=-1}$ in terms of $\digammafcn^{(-2)}(x)$.

\subsection{Elliptic Fekete points in the interval $[-1,1]$}
Regarding the external field problem associated with relation \eqref{Jacobi.expr}, we are interested in the asymptotic expansion of the minimum value of the {\em potential energy}
\begin{equation} \label{eq:jacobi.potential.energy}
\mathcal{L}([-1,1], q, p; x_1, \dots, x_n) \DEF 2 \log \frac{1}{T_n(x_1,\dots,x_n)}, \qquad x_1, \dots, x_n \in [-1,1],
\end{equation}
as $n \to \infty$. 
An $n$-point configuration $\{x_{1}^{(n)}, \dots, x_{n}^{(n)} \}$ minimizing \eqref{eq:jacobi.potential.energy}, or equivalently, maximizing \eqref{Jacobi.expr} over all $n$-point configurations in $[-1,1]$ is called an {\em elliptic $(p,q)$-Fekete $n$-point configuration} in $[-1,1]$ associated with the external field implied by \eqref{Jacobi.expr}. We remark that taking twice of $\log(1/T_n)$ as the potential energy is consistent with the physicist's point of view that the potential energy contained in the electrostatic field of $N$ charges $q_1,\dots,q_N$ at positions $z_1,\dots,z_N$ in the plane, up to some constant factor arising from the used unit system, is given by $\sum_{j\neq k} q_j q_k \log(1/|z_j-z_k|)$; see, e.g., Jackson~\cite{Ja1998}.

\begin{thm} \label{thm:Jacobi}
Let $p > 0$ and $q > 0$. The potential energy of elliptic $(p,q)$-Fekete $n$-point configurations in the interval $[-1,1]$ has the Poincar{\'e}-type asymptotic expansion
\begin{equation*}
\begin{split}
\mathcal{L}([-1,1], q, p; n) 
&= \left( \log 2 \right) n^2 - n \log n + 2 \left( \log 2 \right) \left( p + q - 1 \right) n - 2 \left( \left( p - \frac{1}{4} \right)^2 + \left( q - \frac{1}{4} \right)^2 \right) \log n \\
&\phantom{=}+ C_1(p,q) + \sum_{m=1}^{\infty} \frac{(-1)^{m-1}}{m\left(m+1\right)} \, \mathcal{H}_m(p,q) \, n^{-m},
\end{split}
\end{equation*}
where 
\begin{align*}
C_1(p,q) 
&\DEF 2 \left( \left( p + q - 1 \right)^2 - \frac{11}{24} \right) \log 2 - \left( p + q \right) \log \pi - 3 \log A + \digammafcn^{(-2)}( 2p ) + \digammafcn^{(-2)}( 2q ), \\
\mathcal{H}_m(p,q) 
&\DEF \zetafcn( -m-1 ) + \zetafcn( -m-1, 2p ) + \zetafcn( -m-1, 2q ) + \left( 1 - 2^{-m} \right) \zetafcn( -m-1, 2p + 2q - 1 ).
\end{align*}
\end{thm}

\begin{rmk}
The potential energy of elliptic $(p,q)$-Fekete $n$-point configurations on the interval $[-1,1]$ is invariant under translation (and rotation) of the line-segment $[-1,1]$ in the complex plane. From \eqref{Jacobi.expr} it can be seen that for a scaling constant $\eta>0$ there holds
\begin{equation*}
\mathcal{L}(\eta [-1,1], p , q; n) = \mathcal{L}([-1,1], p , q; n) - \left( \log \eta \right) n^2 - \left( \log \eta \right) \left( 2 p + 2 q -1 \right) n.
\end{equation*}
Thus, only the $n^2$-term and $n$-term are sensitive to a rescaling of the underlying interval.
\end{rmk}

\begin{rmk}
The $n$-th $(p,q)$-discriminant of the interval $[-1,1]$ is given by (cf. Proof of Theorem~\ref{thm:Jacobi})
\begin{equation*}
\Delta_n^{(p,q)}([-1,1]) = 2^{n ( n + 2p + 2q - 1 )} \frac{\prod_{k=1}^n  k^k \left( k + 2p - 1 \right)^{k+2p-1} \left( k + 2q - 1 \right)^{k+2q-1}}{\prod_{k=n-1}^{2(n-1)} \left( k + 2p + 2q \right)^{k+2p+2q}}
\end{equation*}
from which follows an explicit formula for $\mathcal{L}([-1,1], q, p; n)$. An explicit formula in terms of quantities related to Jacobi polynomials is given in \eqref{eq:explicit.cal.L.interval.p.q.n}.
\end{rmk}

In the symmetric external field case $p = q$ we have the following result.
\begin{cor} \label{cor:Jacobi}
Let $p>0$. The potential energy of elliptic $(p,p)$-Fekete $n$-point configurations in the interval $[-1,1]$ has the Poincar{\'e}-type asymptotic expansion
\begin{equation*}
\begin{split}
\mathcal{L}([-1,1], p, p; n) 
&= \left( \log 2 \right) n^2 - n \log n + 2 \left( \log 2 \right) \left( 2 p - 1 \right) n - 4 \left( p - \frac{1}{4} \right)^2 \log n + C_1(p) \\
&\phantom{=}+ \sum_{m=1}^{\infty} \frac{(-1)^{m-1}}{m \left( m + 1 \right)} \,  \mathcal{H}_m(p) \, n^{-m},
\end{split}
\end{equation*}
where 
\begin{align*}
C_1(p) 
&\DEF 2 \left( \left( 2 p - 1 \right)^2 - \frac{11}{24} \right) \log 2 - 2 p \log \pi - 3 \log A + 2 \digammafcn^{(-2)}( 2p ), \\
\mathcal{H}_m(p,q) 
&\DEF \zetafcn( -m-1 ) + 2 \zetafcn( -m-1, 2p ) + \left( 1 - 2^{-m} \right) \zetafcn( -m-1, 4p - 1 ).
\end{align*}
\end{cor}

The asymptotic expansion of the {\bf logarithmic energy} of elliptic $(p,q)$-Fekete $n$-point configurations in $[-1,1]$ is given next.
\begin{thm} \label{thm:Jacobi.log.energy}
Let $p > 0$ and $q > 0$. The logarithmic energy of elliptic $(p,q)$-Fekete $n$-point configurations $\omega_n$ in $[-1,1]$ has the Poincar{\'e}-type asymptotic expansion
\begin{equation*}
\begin{split}
E_0(\omega_n) &= \left( \log 2 \right) n^2 - n \log n - 2 \left( \log 2 \right) n + 2 \left( p^2 + q^2 - \frac{1}{8} \right) \log n + C_1^\prime(p,q) \\
&\phantom{=}+ \sum_{m=1}^{\infty} \frac{(-1)^{m-1}}{m} \, \mathcal{H}_{m}^\prime(p,q) \, n^{-m}
\end{split}
\end{equation*}
as $n \to \infty$, where
\begin{align*}
C_1^\prime(p,q) 
&\DEF - 2 \left( \left( p + q \right)^2 - \frac{13}{24} \right) \log 2 - 3 \log A - 2 p \log \gammafcn(2p) + \digammafcn^{(-2)}(2p) - 2 q \log \gammafcn(2q) + \digammafcn^{(-2)}(2q), \\
\mathcal{H}_{m}^\prime(p,q)
&\DEF \frac{\zetafcn( -m-1 ) + \zetafcn( -m-1, 2p ) + \zetafcn( -m-1, 2q ) + \left( 1 - 2^{-m} \right) \zetafcn( -m-1, 2p+2q-1)}{m+1} \\
&\phantom{=}- 2p \zetafcn( -m, 2p ) - 2q \zetafcn( -m, 2q ) - 2 \left( 1 - 2^{-m} \right) \left( p + q \right) \zetafcn( -m, 2p + 2q - 1 ).
\end{align*}
\end{thm}

\begin{rmk}
Note that the asymptotic expansions of the potential and the logarithmic energy of elliptic $(p,q)$-Fekete $n$-point configurations $\omega_n$ in $[-1,1]$ coincide in the first two leading terms if $p+q\neq 2$ and coincide in the first three leading terms if $p+q=2$.
\end{rmk}

In the case $p = q = 1$, maximizing relation \eqref{Jacobi.expr} for $n$-point configurations in the interval $[-1,1]$ is equivalent with maximizing the product of all mutual distances of $N=n+2$ points in $[-1,1]$: 
\begin{equation}
\mathop{\prod_{j=0}^{n+1}\prod_{k=0}^{n+1}}_{j\neq k} | x_j - x_k |, \qquad -1 \leq x_0, x_1, \dots, x_n, x_{n+1} \leq 1.
\end{equation}
(Indeed, if an endpoint of the interval $[-1,1]$ is not in a configuration $\omega_N$, then the product of all mutual distances between points can be increased by rescaling the points in $\omega_N$.) Hence, the elliptic $(1,1)$-Fekete $n$-point configuration in $[-1,1]$ together with the endpoints $\pm 1$ is also a Fekete $N$-point configuration $\omega_N^*$ on the interval $[-1,1]$ with $N = n + 2$ points. From the electrostatic interpretation of the zeros of classical orthogonal polynomials (cf. Theorem~\ref{thm:Jacobi.electrostatic.interpretation} and remark after that theorem), we have that $\omega_N^*$ is the set of all extremal points (including endpoint extremas) of the Legendre polynomial $\LegendreP_{n+1} = \LegendreP_{N-1}$. 

\begin{thm} \label{thm:log.n.point.energy.asymptotics}
The logarithmic $N$-point energy of the interval $[-1,1]$ has the Poincar{\'e}-type asymptotic expansion
\begin{equation*}
\begin{split}
\mathcal{E}_0([-1,1]; N) &= \left( \log 2 \right) N^2 - N \log N - 2 \left( \log 2 \right) N - \frac{1}{4} \log N + \frac{13 \log 2}{12} - 3 \log A \\
&\phantom{=}+ \sum_{m=1}^{\infty} \frac{1}{m(m+1)} \left(  1 - 2^{-m} + 4 \left( 1 - 2^{-(m+2)} \right) \frac{B_{m+2}}{m+2} \right) N^{-m}
\end{split}
\end{equation*}
as $N\to\infty$. Here, $A$ denotes the Glaisher-Kinkelin constant given in \eqref{eq:Glaisher-Kinkelin}.
\end{thm}

\begin{rmk}
The $N$-th discriminant of the interval $[-1,1]$ defined in \eqref{eq:Nth.discriminant.of.interval} can be written as (cf. Proof of Theorem~\ref{thm:log.n.point.energy.asymptotics})
\begin{equation*}
\Delta_N([-1,1]) = 2^{N(N-1)} N^N \frac{\prod_{k=1}^{N-1} k^{3k}}{\prod_{k=N-1}^{2(N-1)} k^{k}}
\end{equation*}
and via \eqref{eq:min.log.energy} we get an explicit formula for $\mathcal{E}_0([-1,1]; N)$. An explicit formula in terms of quantities related to Jacobi polynomials is given in \eqref{eq:min.log.energy.of.interval}.
\end{rmk}

\subsection{Fekete points in the interval $[-2,2]$}

This case has been treated analytically in \cite{Po1964}. More generally, Pommerenke obtained that for a convex compact planar set $A$ of transfinite diameter (logarithmic capacity) $\CAP A$, the $N$-th discriminant of $A$ sastisfies
\begin{equation*}
N^N \left( \CAP A \right)^{N(N-1)} \leq \Delta_N(A) \leq 2^{2(N-1)} N^N \left( \CAP A \right)^{N(N-1)}.
\end{equation*}
Let $W(A) \DEF - \log ( \CAP A )$ denote the {\em logarithmic energy of $A$}. Then is follows that the logarithmic $N$-point energy of convex compact planar set $A$ satisfies
\begin{equation*}
\left( W(A) - \log 4 \right) N + \log 4 \leq \mathcal{E}_0(A; N) - \left( W(A) N^2 - N \log N \right) \leq W(A) N.
\end{equation*}
Considering the star-shaped curves $S_m = \bigcup_{\nu=1}^m [ 0, 2^{2/m} \zeta^\nu]$ ($\zeta \DEF e^{2\pi i / m}$) of transfinite diameter $1$ defined by the conformal mapping $F(z) = z ( 1 + z^{-m} )^{2/m}$, where $m$ is the number of star branches, he showed that $\Delta_N(S_2) \geq 2^{2(N-1)} N^N$. Consequently, for $A=[-2,2]=S_2$ these results imply 
\begin{equation*}
2^{2(N-1)} N^N \leq \Delta_N(S_2) \leq 2^{2(N-1)} N^N.
\end{equation*}

In \cite{BeClDu2004} the electrostatic equilibria of $N$ discrete charges of size $1/N$ on a two-dimensional conductor (domain) are studied. Also \cite{BeClDu2004} is mostly concerned with placement of charges, it provides an interpretation of the terms of the asymptotics of the ground-state energy, which we will follow here. From Theorem~\ref{thm:log.n.point.energy.asymptotics} we have that (note that $\CAP[-2,2]=1$ and therefore $W([-2,2])=0$)
\begin{align*}
\frac{\mathcal{E}_0([-2,2];N)}{N^2} 
&= W([-2,2]) & &\text{(continuum correlation energy)} \\
&\phantom{=}- \frac{\log N}{N} & &\text{(self energy)} \\
&\phantom{=}- \frac{\log 2}{N} & &\text{(correlation energy)} \\
&\phantom{=}- \frac{1}{4} \frac{\log N}{N^2} & & \\
&\phantom{=}- \left( \frac{13 \log 2}{12} - 3 \log A \right) \frac{1}{N^2} & & \\
&\phantom{=}+ \cdots, & &
\end{align*}
where $\log A$ is the logarithm of the Glaisher-Kinkelin constant, see \eqref{eq:Glaisher-Kinkelin}. In fact, Theorem~\ref{thm:log.n.point.energy.asymptotics} gives the complete asymptotic expansion of $\mathcal{E}_0([-1,1];N)$ as $N\to\infty$. Note that only the $N^2$-term and $(\log N)$-term are affected by a change of the transfinite diameter; i.e., as $N \to \infty$:
\begin{equation*}
\begin{split}
\mathcal{E}_0([a,b]; N) &= W([a,b]) \, N^2 - N \log N - \left( \log 2 + W([a,b]) \right) N - \frac{1}{4} \log N + \frac{13 \log 2}{12} - 3 \log A \\
&\phantom{=}+ \sum_{m=1}^{\infty} \frac{1}{m(m+1)} \left(  1 - 2^{-m} + 4 \left( 1 - 2^{-(m+2)} \right) \frac{B_{m+2}}{m+2} \right) N^{-m}.
\end{split}
\end{equation*}

\section{Asymptotics of the discriminant of the Jacobi polynomial} 
\label{sec:aux.res}

For the proof of Theorem~\ref{thm:Jacobi} we need an asymptotic expansion of the leading coefficient, the values at $\pm 1$, and the discriminant of the Jacobi polynomial. We recall the following facts. The Jacobi polynomials $P_n^{(\alpha, \beta)}(x)$ ($n\geq0$, $\alpha, \beta > -1$) are orthogonal on the interval $[-1,1]$ with the weight function $w(x) = ( 1 - x )^\alpha ( 1 + x )^\beta$ and normalized such that $P_n^{(\alpha,\beta)}(1) = \Pochhsymb{1+\alpha}{n} / n!$. Hence
\begin{equation*}
P_n^{(\alpha,\beta)}(x) = \lambda_n^{(\alpha,\beta)} x^n + \cdots, \qquad \lambda_n^{(\alpha,\beta)} = 2^{-n} \binom{2n+\alpha+\beta}{n}.
\end{equation*}
We note further that $P_n^{(\alpha,\beta)}(-x) = (-1)^n P_n^{(\beta,\alpha)}(x)$. Therefore, $P_n^{(\alpha,\beta)}(-1) = (-1)^n \Pochhsymb{1+\beta}{n} / n!$.

We prove the following Poincar{\'e}-type asymptotic results expressed in terms of the zeta function and the Hurwitz zeta function.

\begin{lem} \label{lem:Jacobi.asymptotics}
Let $\alpha > -1$ and $\beta > -1$. Then 
\begin{align*}
\log \lambda_n^{(\alpha,\beta)} &= \left( \log 2 \right) n - \frac{1}{2} \log n + \left( \alpha + \beta \right) \log 2 - \frac{1}{2} \log \pi \\
&\phantom{=}+ \sum_{m=1}^{\infty} \frac{(-1)^{m-1}}{m} \Big( \left( 1 - 2^{-m} \right) \zetafcn( -m, \alpha + \beta + 1 ) + \zetafcn( - m ) \Big) n^{-m}, \\
\log P_n^{(\alpha,\beta)}(1)
&= \alpha \log n - \log \gammafcn(\alpha+1) + \sum_{m=1}^{\infty} \frac{(-1)^{m}}{m} \, \Big( \zetafcn( -m, \alpha + 1 ) - \zetafcn( -m ) \Big) \, n^{-m}.
\end{align*}
\end{lem}

\begin{proof}
Since 
\begin{equation*}
\log \lambda_n^{(\alpha,\beta)} = - n \log 2 + \log \gammafcn( 2n + \alpha + \beta + 1 ) - \log \gammafcn( n + \alpha + \beta + 1 ) - \log \gammafcn(n + 1),
\end{equation*}
application of \eqref{eq:LogGamma.asymptotics} and simplification gives the first result. 

For the second part we have
\begin{equation*}
\log P_n^{(\alpha,\beta)}(1) = - \log \gammafcn(\alpha+1) + \log \gammafcn(n+\alpha+1) - \log \gammafcn(n+1)
\end{equation*}
and application of \eqref{eq:LogGamma.asymptotics} yields the second part. 

In either part we used $\zetafcn(-m,1) = \zetafcn(-m)$ for $m \geq 1$.  
\end{proof}

The connection between the energy optimization problem and the zeros of certain Jacobi polynomials is established in the following theorem. Uniqueness of the maximal configuration also follows from this fact.
\begin{thm}[{\cite[Thm.~6.7.1]{Sz1939}}] \label{thm:Jacobi.electrostatic.interpretation}
Let $p > 0$ and $q > 0$, and let $\{x_1, \dots, x_n\}$ be a system of real numbers in the interval $[-1,1]$ for which the expression \eqref{Jacobi.expr} becomes a maximum. Then $x_1, \dots, x_n$ are the zeros of the Jacobi polynomial $P_n^{(\alpha, \beta)}(x)$, where $\alpha = 2 p - 1$, $\beta = 2q - 1$. 
\end{thm}

\begin{rmk}
In the particular case of $p = q = 1$, it follows from the well-known relations (cf. \cite[Ch.~18]{DLMF2021.06.15})
\begin{equation*}
P_n^{(1,1)}( x ) = \frac{2}{n+2} \, \GegenbauerC_n^{(3/2)}( x ) = \frac{2}{n+2} \, \frac{\dd \LegendreP_{n+1}}{\dd x}(x)
\end{equation*}
that the unique maximizing configuration for \eqref{Jacobi.expr} in the interval $[-1,1]$ can be characterized as the set of the zeros of the Jacobi polynomial $P_n^{(1,1)}$, the zeros of the Gegenbauer polynomial $\GegenbauerC_n^{(3/2)}$, or the extremas of the Legendre polynomial $\LegendreP_{n+1}$.
\end{rmk}


An explicit formula for the discriminant of $P_n^{(\alpha,\beta)}(x) = \lambda_n^{(\alpha,\beta)} ( x - x_{1,n} ) \cdots ( x - x_{n,n} )$, defined by 
\begin{equation} \label{eq:D.n.alpha.beta}
D_n^{(\alpha,\beta)} \DEF \left[ \lambda_n^{(\alpha,\beta)} \right]^{2n-2} \mathop{\prod_{j=1}^n \prod_{k=1}^n}_{j<k} \left( x_{j,n} - x_{k,n} \right)^2,
\end{equation}
can be obtained without computing the zeros of Jacobi polynomials:

\begin{thm}[{\cite[Thm.~6.71]{Sz1939}}] \label{thm:Jacobi.electrostatic.interpretation.discriminant} Let $\alpha>-1$ and $\beta>-1$. Then 
\begin{align*}
D_n^{(\alpha,\beta)} 
&= 2^{-n(n-1)} \prod_{\nu=1}^n \nu^{\nu-2n+2} \left( \nu + \alpha \right)^{\nu-1} \left( \nu + \beta \right)^{\nu-1} \left( \nu + n + \alpha + \beta \right)^{n-\nu}.
\end{align*}
\end{thm}

The logarithm of the discriminant of the Jacobi polynomials admits the following Poincar{\'e}-type asymptotic expansion.
The Glaisher-Kinkelin constant $A$ is given in \eqref{eq:Glaisher-Kinkelin} and the negapolygamma function $\digammafcn^{(-2)}$ is given in \eqref{eq:negadigamma}.
\begin{lem} \label{lem:Jacobi.discr.asymptotics}
Let $\alpha > -1$ and $\beta > -1$. Then for every integer $K \geq1$ there holds
\begin{equation*}
\begin{split}
\log D_n^{(\alpha,\beta)} &= (\log 2 ) n^2 + \left( 2 \left( \alpha + \beta \right) \log 2 - \log \pi¸ \right) n + \frac{1}{2} \left( \frac{5}{2} - \left( \alpha + 1 \right)^2 - \left( \beta + 1 \right)^2 \right) \log n + C(\alpha,\beta) \\
&\phantom{=}- \sum_{m=1}^{\infty} \frac{(-1)^{m-1}}{m} \, \Psi_m(\alpha, \beta) \, n^{-m},
\end{split}
\end{equation*}
where
\begin{align*}
C(\alpha,\beta) 
&\DEF - \frac{1}{8} - \frac{1}{2} \left( \alpha + \beta + \frac{1}{2} \right)^2 + \frac{1}{2} \left( \frac{11}{6} + \left( \alpha + \beta \right)^2 \right) \log 2 + \log \pi + 3 \log A  \\
&\phantom{\DEF}+ \left( \alpha + 1 \right) \log \gammafcn(\alpha + 1) - \digammafcn^{(-2)}( \alpha + 1 ) + \left( \beta + 1 \right) \log \gammafcn(\beta + 1) - \digammafcn^{(-2)}( \beta + 1 ), \\
\Psi_m(\alpha, \beta) 
&\DEF - \frac{2m+1}{m+1} \zetafcn( -m - 1 ) - 2 \zetafcn( - m ) \\
&\phantom{\DEF}+ \left( \alpha + 1 \right) \zetafcn( -m, \alpha + 1 ) - \frac{\zetafcn( -m-1, \alpha + 1 )}{m+1} + \left( \beta + 1 \right) \zetafcn( -m, \beta + 1 ) - \frac{\zetafcn( -m-1, \beta + 1 )}{m+1} \\
&\phantom{\DEF}- \frac{\left( 2 - 2^{-m} \right) m + 1 - 2^{-m}}{m+1} \zetafcn( -m-1, \alpha + \beta + 1 ) + \left( \alpha + \beta \right) \left( 1 - 2^{-m} \right) \zetafcn( -m, \alpha + \beta + 1 ).
\end{align*}
\end{lem}

\begin{proof}
First, we observe that differentiating the identity
\begin{equation*}
\sum_{k=m+1}^n \left( k + x + a \right)^{-s} = \zetafcn(s, m + x + a + 1 ) - \zetafcn(s, n + x + a + 1), \qquad 0 \leq m < n,
\end{equation*}
with respect to $s$ and setting $s=-1$ gives the following formula (using $\zetafcn^\prime( -1, z ) := \frac{\partial}{\partial s} \zetafcn( s, z ) |_{s = - 1}$)
\begin{equation} \label{eq:logsum2}
\sum_{k=m+1}^n \left( k + x + a \right) \log ( k + x + a ) = \zetafcn^\prime(-1, n + x + a + 1 ) - \zetafcn^\prime(-1, m + x + a + 1 ), \qquad 0 \leq m < n.
\end{equation}
Hence
\begin{equation*}
\log D_n^{(\alpha,\beta)} = - n \left( n - 1 \right) \log 2 + \mathfrak{A}_n + \mathfrak{B}_n(\alpha) + \mathfrak{B}_n(\beta) + \mathfrak{C}_n(\alpha+\beta),
\end{equation*}
where for $\alpha>-1$ and $b>-2$:
\begin{align*}
\mathfrak{A}_n &\DEF \sum_{\nu=1}^n \left( \nu - 2 n + 2 \right) \log \nu = \zetafcn^\prime(-1, n + 1 ) - \zetafcn^\prime(-1) - 2 \left( n - 1 \right) \log \gammafcn(n+1), \\
\mathfrak{B}_n(\alpha) &\DEF \sum_{\nu=1}^n \left( \nu - 1 \right) \log (\nu + \alpha) = \zetafcn^\prime(-1, n + \alpha + 1 ) - \zetafcn^\prime(-1, \alpha + 1 ) - \left( \alpha + 1 \right) \log \Pochhsymb{\alpha+1}{n}, \\
\mathfrak{C}_n(b) &\DEF \sum_{\nu=1}^n \left( n - \nu \right) \log ( \nu + n + b ) = \left( 2 n + b \right) \log \Pochhsymb{n+b+1}{n} - \zetafcn^\prime(-1, 2 n + b + 1 ) + \zetafcn^\prime(-1, n + b + 1 ).
\end{align*}

The asymptotic forms follow from applying \eqref{eq:LogGamma.asymptotics} and \eqref{eq:1st.HurwitzZeta.s.derivative.at.minus.1}. Simplification is done with the help of Mathematica.

First, we get the Poincar{\'e}-type asymptotics 
\begin{equation*}
\begin{split}
\mathfrak{A}_n 
&= - \frac{3}{2} n^2 \log n + \frac{7}{4} n^2 + \frac{3}{2} n \log n - \left( 2 + \log( 2 \pi ) \right) n + \frac{13}{12} \log n + \log A - \frac{1}{6} + \log( 2 \pi ) \\
&\phantom{=}+ \sum_{m=1}^\infty \frac{(-1)^{m}}{m} \left( 2 \zetafcn( -m ) + \frac{2m + 1}{m+1} \zetafcn( -m - 1 ) \right) n^{-m}, 
\end{split}
\end{equation*}
where $A$ is the {\em Glaisher-Kinkelin} constant. We used $\zetafcn^\prime( - 1 ) = \frac{1}{12} - \log A$.

Furthermore,
\begin{equation*}
\begin{split}
\mathfrak{B}_n(\alpha) 
&= \frac{1}{2} n^2 \log n - \frac{1}{4} n^2 - \frac{1}{2} n \log n + \left( \alpha + 1 \right) n + \frac{1}{2} \left( \frac{1}{6} - \left( \alpha + 1 \right)^2 \right) \log n \\
&\phantom{=}+ \log A - \digammafcn^{(-2)}( \alpha + 1 ) + \left( \alpha + 1 \right) \log \gammafcn( \alpha + 1 )  \\
&\phantom{=}+ \sum_{m=1}^\infty \frac{(-1)^{m-1}}{m} \left( \left( \alpha + 1 \right) \zetafcn( -m, \alpha + 1 ) - \frac{\zetafcn( -m - 1, \alpha + 1 )}{m+1}  \right) n^{-m}.
\end{split}
\end{equation*}
Here, we used the negapolygamma function defined in \eqref{eq:negadigamma} to simplify the constant term.

Furthermore,
\begin{equation*}
\begin{split}
\mathfrak{C}_n(b)
&= \frac{1}{2} n^2 \log n + \left( 2 \log 2 - \frac{5}{4} \right) n^2 - \frac{1}{2} n \log n + \left( 2 \log 2 - 1 \right) b \, n + \frac{1}{2} \left( b^2 - \frac{1}{6} \right) \log 2 - \frac{1}{2} \left( b \left( b + 1 \right) + \frac{1}{6} \right) \\
&\phantom{=}+ \sum_{m=1}^\infty \frac{(-1)^m}{m} \left( \frac{ \frac{2 - 2^{-m}}{1 - 2^{-m}} \, m + 1}{m+1} \zetafcn( -m - 1, b + 1 ) - b \zetafcn( -m, b + 1 ) \right) \left( 1 - 2^{-m} \right)  n^{-m}.
\end{split}
\end{equation*}

Putting everything together, we arrive at the desired result. 
\end{proof}


\section{Proofs of main results}
\label{sec:proofs.main.res}


\begin{proof}[Proof of Theorem~\ref{thm:Jacobi}]
By Theorem~\ref{thm:Jacobi.electrostatic.interpretation}, the elliptic $(p,q)$-Fekete $n$-point configuration in $[-1,1]$ is give by the zeros of the Jacobi polynomial  $P_n^{(\alpha,\beta)}$ for $\alpha = 2 p - 1$ and $\beta = 2 q - 1$. 
We set $\alpha = 2 p - 1$ and $\beta = 2 q - 1$.
Let $x_{1,n},\dots, x_{n,n}$ denote the $n$ zeros of $P_n^{(\alpha,\beta)}$. 
From \eqref{eq:D.n.alpha.beta} and Theorem~\ref{thm:Jacobi.electrostatic.interpretation.discriminant} it follows that
\begin{equation} \label{eq:T.n.Jacobi}
T_n(x_{1,n},\dots, x_{n,n}) = \frac{\left[ P_n^{(\alpha,\beta)}(1) \right]^p}{\left[ \lambda_n^{(\alpha,\beta)} \right]^p}  \frac{\sqrt{D_n^{(\alpha,\beta)}}}{\left[ \lambda_n^{(\alpha,\beta)} \right]^{n-1}} \frac{\left[ (-1)^n P_n^{(\alpha,\beta)}(-1) \right]^q}{\left[ \lambda_n^{(\alpha,\beta)} \right]^q}
\end{equation}
and therefore (recall, $\alpha = 2p-1$ and $\beta = 2q-1$)
\begin{equation} \label{eq:explicit.cal.L.interval.p.q.n}
\begin{split}
\mathcal{L}([-1,1], q, p; n) 
&= 2 \left( n + p + q - 1 \right) \log \lambda_n^{(\alpha,\beta)} - \log D_n^{(\alpha,\beta)} \\
&\phantom{=}- 2p \, \log P_n^{(\alpha,\beta)}(1) - 2q \, \log P_n^{(\beta,\alpha)}(1).
\end{split}
\end{equation}
Utilizing Lemma~\ref{lem:Jacobi.asymptotics} and Lemma~\ref{lem:Jacobi.discr.asymptotics}, we get the desired result with the help of Mathematica. 
\end{proof}

\begin{proof}[Proof of Theorem~\ref{thm:Jacobi.log.energy}]
Recall that $\alpha=2p-1$ and $\beta=2q-1$. From \eqref{eq:D.n.alpha.beta} we obtain
\begin{equation} \label{eq:E.0.Jacobi.aux}
E_0(x_{1,n}, \dots, x_{n,n}) = 2 \left( n - 1 \right) \log \lambda_n^{(\alpha,\beta)} - \log D_n^{(\alpha,\beta)}.
\end{equation}
Utilizing Lemma~\ref{lem:Jacobi.asymptotics} and Lemma~\ref{lem:Jacobi.discr.asymptotics}, we get the desired result with the help of Mathematica. 
\end{proof}

\begin{proof}[Proof of Theorem~\ref{thm:log.n.point.energy.asymptotics}]
Suppose $p>0$ and $q>0$. Set $\alpha = 2p-1$ and $\beta = 2q-1$. Let $\omega_n=\{x_{1,n},\dots,x_{n,n}\}$ be an elliptic $(p,q)$-Fekete $n$-point configurations in $[-1,1]$. Rewriting \eqref{eq:log.energy} and using Theorem~\ref{thm:Jacobi.electrostatic.interpretation}, we get
\begin{align*}
E_0(\omega_n \cup \{-1, +1\}) 
&= E_0(\omega_n) + 2 \sum_{k=1}^n \log \frac{1}{\left| -1 - x_{k,n} \right|} + 2 \sum_{k=1}^n \log \frac{1}{\left| 1 - x_{k,n} \right|} + 2 \log \frac{1}{\left| - 1 - 1 \right|} \\
&= E_0(\omega_n) - 2 \log \left| \prod_{k=1}^n \left( -1 -x_{k,n} \right) \right| - 2 \log \left| \prod_{k=1}^n \left( 1 -x_{k,n}  \right) \right| - 2 \log 2 \\
&= E_0(\omega_n) - 2 \log \left| \frac{P_n^{(\alpha,\beta)}(-1)}{\lambda_n^{(\alpha,\beta)}} \right| - 2 \log \left| \frac{P_n^{(\alpha,\beta)}(1)}{\lambda_n^{(\alpha,\beta)}} \right| - 2 \log 2 \\
&= 2 \left( n + 1 \right) \log \lambda_n^{(\alpha,\beta)} - \log D_n^{(\alpha,\beta)} - 2 \log P_n^{(\beta,\alpha)}(1) - 2 \log P_n^{(\alpha,\beta)}(1) - 2 \log 2.
\end{align*}
The substitution for $E_0(\omega_n)$ follows from \eqref{eq:E.0.Jacobi.aux}. 

For $p = q = 1$ and $n = N - 2$, we get 
\begin{equation} \label{eq:min.log.energy.of.interval}
\mathcal{E}_0([-1,1]; N) = 2 \left( N - 1 \right) \log \lambda_{N-2}^{(1,1)} - \log D_{N-2}^{(1,1)} - 4 \log P_{N-2}^{(1,1)}(1) - 2 \log 2.
\end{equation}
The asymptotic expansions of Lemma~\ref{lem:Jacobi.asymptotics} and Lemma~\ref{lem:Jacobi.discr.asymptotics} are not in terms of the new asymptotic variable $N$. Instead we simplify the right-hand side above further and use ideas from the proof of Lemma~\ref{lem:Jacobi.discr.asymptotics}. Combining the logarithmic terms and simplification yields
\begin{equation*}
\mathcal{E}_0([-1,1]; N) = E_0(\omega_n \cup \{-1, +1\}) = - \log \Delta_N([-1,1]),
\end{equation*}
where
\begin{equation*}
\Delta_N([-1,1]) = 2^{N(N-1)} N^N \left( \prod_{k=1}^{N-1} k^{3k} \right) \left( \prod_{k=N-1}^{2(N-1)} k^{-k} \right).
\end{equation*}
Hence, using \eqref{eq:logsum2},
\begin{equation*}
\begin{split}
\mathcal{E}_0([-1,1]; N) 
&= - N ( N - 1 ) \log 2 - N \log N + 3 \zetafcn^\prime( -1, 1 ) \\
&\phantom{=}- 3 \zetafcn^\prime( -1, N ) - \zetafcn^\prime( -1, N - 1 ) + \zetafcn^\prime( -1, 2N-1 ). 
\end{split}
\end{equation*}
Application of \eqref{eq:1st.HurwitzZeta.s.derivative.at.minus.1} and simplification gives the desired result. We used the following simplification (taking into account that Bernoulli numbers $B_k$ with odd integers $k \geq 3$ vanish):
\begin{align*}
\left(1 - 2^{-m} \right) \zetafcn( -m-1 -1 ) + 3 \zetafcn( -m-1, 0 ) 
&= - \left(1 - 2^{-m} \right) \frac{\bernoulliB_{m+2}(-1)}{m+2} - 3 \frac{B_{m+2}}{m+2} \\
&= (-1)^{m-1} \left(  1 - 2^{-m} + 4 \left( 1 - 2^{-(m+2)} \right) \frac{B_{m+2}}{m+2} \right).
\end{align*}
\end{proof}

\appendix

\section{Basic asymptotic expansions}

\subsection{Gamma function asymptotics}

We use the Poincar{\'e}-type asymptotics (cf., eg., \cite[Eq.~5.11.8]{DLMF2021.06.15})
\begin{equation} \label{eq:LogGamma.asymptotics}
\log \gammafcn(x + a) = \left( x + a - 1 / 2 \right) \log x - x + \frac{1}{2} \log( 2 \pi ) - \sum_{m=1}^{\infty} \frac{(-1)^{m-1}}{m} \, \zetafcn( - m, a ) \, x^{-m}
\end{equation}
as $x \to \infty$ and $a \in \mathbb{R}$ fixed. Observe that
\begin{equation*}
\zetafcn( -m, a ) = - \frac{\bernoulliB_{m+1}( a )}{m+1}, \qquad m \in \mathbb{N}_0.
\end{equation*}

\subsection{$s$-derivative of the Hurwitz zeta function}

We need the asymptotic expansion as $n \to \infty$ of $\frac{\partial}{\partial s}\zetafcn(s,n+a) |_{s=-1}$. Here, the asymptotic variable is shifted by a fixed (positive) real number. In case of $a = 0$, see, e.g., \cite[§25.11(xii)]{DLMF2021.06.15}. In case of $a > 0$, cf. \cite{Ka1998, Ka2020}. 

The Mellin-Barnes integral approach gives
\begin{equation} \label{eq:Hurwitz.Zeta.shifted}
\zetafcn( s, z + a ) = \frac{z^{1-s}}{s-1} + \sum_{k=0}^{K-1} \binom{-s}{k} \, \zetafcn( -k, a ) \, z^{-s-k} + \rho_K( s, a, z)
\end{equation}
valid in the sector $|\arg z| < \pi$ and such that $\re s > - K$ and $\alpha > 0$. The remainder term takes the form
\begin{equation}  \label{eq:Hurwitz.Zeta.shifted.remainder}
\rho_K( s, a, z) = \frac{1}{2\pi \, i} \int_{\gamma_K - i \infty}^{\gamma_K + i \infty} \frac{\gammafcn(w+s) \gammafcn(-w)}{\gammafcn(s)} \zetafcn(w+s, \alpha) \, z^w \dd w,
\end{equation}
where $\gamma_K$ satisfies $-1 - \re s - K < \gamma_K < - \re s - K$. Throughout the sector $| \arg z | \leq \pi - \delta < \pi$ holds the estimate $\rho_K( s, a, z) = \mathcal{O}( | z |^{-\re s - K} )$ as $z \to \infty$.

Partial differentiation with respect to $s$ yields
\begin{equation}
\begin{split} \label{eq:1st.HurwitzZeta.s.derivative}
\zetafcn^\prime(s, z + a ) 
&\DEF \frac{\partial }{\partial s} \zetafcn(s, z + a ) 
= - \frac{z^{1-s} \log z}{s-1} - \frac{z^{1-s}}{(s-1)^2} \\
&\phantom{=}- \sum_{k=0}^{K-1} \binom{-s}{k} \, \zetafcn( -k, a ) \, z^{-s-k} \left( \log z + \digammafcn( 1 - s ) - \digammafcn( 1 - s - k ) \right) \\
&\phantom{=} + \rho_K^\prime( s, a, z).
\end{split}
\end{equation}
This formula is valid under the same assumptions as above. It is understood that
\begin{equation*}
\binom{-s}{k} \big( \digammafcn( 1 - s ) - \digammafcn( 1 - s - k ) \big) 
= - \frac{(-1)^k}{k!} \sum_{\ell=0}^{k-1} \frac{\Pochhsymb{s}{k}}{s+\ell}  
= - \frac{(-1)^k}{k!} \sum_{\ell=0}^{k-1}  \Pochhsymb{s}{\ell} \Pochhsymb{s+1+\ell}{k-1-\ell}.
\end{equation*}
Throughout the sector $| \arg z | \leq \pi - \delta < \pi$ holds the estimate $\rho_K( s, a, z) = \mathcal{O}( | z |^{-\re s - K} \log | z | )$ as $z \to \infty$. In particular, after an index shift and for $K \geq 2$, 
\begin{equation}
\begin{split} \label{eq:1st.HurwitzZeta.s.derivative.at.minus.1}
\zetafcn^\prime(-1, z + a ) 
&= \frac{1}{2} x^2 \log x - \frac{1}{4} x^2 - \zetafcn( 0, a ) \, x \log x - \zetafcn( -1, a ) \log x - \zetafcn( -1, a ) \\
&\phantom{=}+ \sum_{k=1}^{K-1} \frac{(-1)^{k}}{k(k+1)} \, \zetafcn( -k-1, a ) \, x^{-k} + \mathcal{O}( x^{-K} \log x ) \qquad \text{as $x \to \infty$.}
\end{split}
\end{equation}
A more detailed analysis shows that the $\log x$ factor in the remainder estimate can be dropped. The remainder term takes the form
\begin{equation*}
\frac{1}{2\pi i} \int_{\gamma_K - i \infty}^{\gamma_K + i \infty} \frac{1}{(w-1)w} \frac{\pi}{\sin( \pi w )} \, \zetafcn( w - 1, a ) \, x^w \, \dd w,
\end{equation*}
where $-1-K < \gamma_K < -K$. 
Furthermore, the restriction $a > 0$ can be relaxed to $a > -1$, $a \neq 0$, by means of the identities $\zetafcn( s, a ) = a^{-s} + \zetafcn( s, a + 1 )$ (and $\bernoulliB_{k+1}( a + 1 ) = \bernoulliB_{k+1}( a ) + ( k + 1 ) a^k$ if Bernoulli polynomials are used). In the case $a = 0$ formula \eqref{eq:1st.HurwitzZeta.s.derivative.at.minus.1} reduces to the well known asymptotic expansion.

\def\cprime{$'$} \def\polhk#1{\setbox0=\hbox{#1}{\ooalign{\hidewidth
  \lower1.5ex\hbox{`}\hidewidth\crcr\unhbox0}}}

\end{document}